\documentclass[a4paper,11pt]{article}
\pdfoutput=1 

\usepackage{jheppub} 

\usepackage[T1]{fontenc} 
\usepackage[english]{babel}
\usepackage{graphicx}			    	
\usepackage{slashed}
\usepackage{caption}
\usepackage{amsmath}
\usepackage{amsthm}
\usepackage{amsfonts}
\usepackage{mathtools}
\usepackage{subcaption}
\usepackage{tikz} 
\usepackage{tikz-feynman}
\usepackage{tikz-3dplot}
\tikzfeynmanset{compat=1.1.0}
\usepackage{placeins}
\usepackage{scalerel}
\usepackage{verbatim}
\usetikzlibrary{shapes.misc}
\tikzset{cross/.style={cross out, draw=black, minimum size=2*(#1-\pgflinewidth), inner sep=0pt, outer sep=0pt},
	cross/.default={5pt}}

\DeclareMathOperator{\Tr}{Tr}

\newcommand{\overbar}[1]{\mkern 1.5mu\overline{\mkern-1.5mu#1\mkern-1.5mu}\mkern 1.5mu}

\theoremstyle{plain}
\newtheorem{theorem}{Theorem}
\newtheorem{prop}{Proposition}
\theoremstyle{plain}
\newtheorem{clr}{Corollary}
\theoremstyle{plain}
\newtheorem{lemma}{Lemma}
\makeatletter
\def\@fpheader{\relax}
\makeatother

\title{Nonresonant renormalization scheme for twist-$2$ operators in $\mathcal{N}=1$  SUSY  SU($N$) Yang-Mills theory}
\author[b,a]{Francesco Scardino}

\affiliation[a]{Physics Department, INFN Roma1, 
	Piazzale A. Moro 2, Roma, I-00185, Italy}
\affiliation[b]{Physics Department, Sapienza University,Piazzale A. Moro 2, Roma, I-00185, Italy}

\emailAdd{francesco.scardino@uniroma1.it}

\abstract{The short-distance asymptotics of the generating functional for $n$-point correlators of twist-$2$ operators in $\mathcal{N}=1$  supersymmetric (SUSY) SU($N$) Yang-Mills (SYM) theory were recently calculated in \cite{BPS41,BPS42}. This calculation depends on a change of basis for renormalized twist-$2$ operators, in which $-\gamma(g)/ \beta(g)$ reduces to $\gamma_0/ (\beta_0\,g)$ at all orders in perturbation theory, where $\gamma_0$ is diagonal, $\gamma(g) = \gamma_0 g^2+\ldots$ is the anomalous-dimension matrix, and $\beta(g) = -\beta_0 g^3+\ldots$ is the beta function. The method is founded on a new geometric interpretation of operator mixing \cite{Bochicchio:2021geometry}, assuming that the eigenvalues of the matrix $\gamma_0/ \beta_0$ meet the nonresonant condition $\lambda_i-\lambda_j\neq 2k$, with the eigenvalues $\lambda_i$ ordered nonincreasingly and $k\in \mathbb{N}^+$. This nonresonant condition was numerically verified for $i,j$ up to $10^4$ in \cite{BPS41,BPS42}. In this work, we employ techniques initially developed in \cite{S1} to present a number-theoretic proof of the nonresonant condition for twist-$2$ operators, fundamentally based on the classic result that Harmonic numbers are not integers.}

\begin{document} 
	
	\definecolor{c969696}{RGB}{150,150,150}
	\maketitle	
	\flushbottom

	\section{Introduction}
	Recently, the ultraviolet (UV) asymptotics of the generating functional of correlators of twist-$2$ operators in SU($N$) SYM theory was explicitly computed for the first time \cite{BPS41,BPS42}. This result establishes strong UV constraints on the anticipated nonperturbative solution of large-$N$ SYM theory and could serve as an essential guide in the search for this solution \cite{BPS41,BPS42}.\par
	Moreover, the aforementioned computation has also led to a refinement of the 't Hooft topological expansion in large-$N$ SU($N$) pure YM theory \cite{BPSL} that is intimately connected with the corresponding nonperturbative effective theory of glueballs \cite{BPSpaper2,BPSL}.\par
	A critical tool for performing this calculation is a change of basis of renormalized twist-$2$ operators, in which the renormalized mixing matrix $Z(\lambda)$ defined in Eq. \eqref{ZZ} becomes diagonal and one-loop exact to all orders of perturbation theory. This is achieved through a new geometric interpretation of operator mixing \cite{Bochicchio:2021geometry}, which we outline below.\par
	As noted in the introduction of \cite{Becchetti:2021for}, a change of renormalization scheme can, in general, involve both a reparametrization of the coupling—which alters the beta function $\beta(g) = -\beta_0 g^3+\ldots$, with $g\equiv g(\mu)$ the renormalized coupling—and a change in the basis of the operators that mix under renormalization, which modifies the anomalous dimension matrix $\gamma(g)=\gamma_0 g^2+ \cdots$. \par
	In this paper, we focus exclusively on a change of the operator basis \cite{Bochicchio:2021geometry}, while holding the renormalization scheme for $\beta(g)$ fixed, for instance, in the $\overbar{\text{MS}}$ scheme.\par
	Naturally, this change of basis also influences the ratio $-\frac{\gamma(g)}{\beta(g)}$, with $\beta(g)$ fixed, in a manner that we will detail shortly.\par
	In the case of operator mixing, the renormalized Euclidean correlators
	\begin{equation}\label{key}
		\langle \mathcal{O}_{k_1}(x_1)\ldots \mathcal{O}_{k_n}(x_n) \rangle = G^{(n)}_{k_1 \ldots k_n}( x_1,\ldots,  x_n; \mu, g(\mu))
	\end{equation}
	satisfy the Callan-Symanzik equation
	\begin{align}\label{CallanSymanzik}
		& \Big(\sum_{\alpha = 1}^n x_\alpha \cdot \frac{\partial}{\partial x_\alpha} + \beta(g)\frac{\partial}{\partial g}+ \sum_{\alpha = 1}^n D_{\mathcal{O}_\alpha}\Big)G^{(n)}_{k_1 \ldots k_n} + \nonumber\\
		&+ \sum_a \Big(\gamma_{k_1a}(g) G^{(n)}_{ak_2 \ldots k_n}+ \gamma_{k_2a}(g) G^{(n)}_{k_1 a k_3 \ldots k_n} \cdots +\gamma_{k_n a}(g) G^{(n)}_{k_1 \ldots a}\Big) = 0\,,
	\end{align}
	with the solution
	\begin{align}\label{csformula}
		&G^{(n)}_{k_1 \ldots k_n}(\lambda x_1,\ldots, \lambda x_n; \mu, g(\mu)) \nonumber\\
		&= \sum_{j_1 \ldots j_n} Z_{k_1 j_1} (\lambda)\ldots Z_{k_n j_n}(\lambda)\hspace{0.1cm} \lambda^{-\sum_{i=1}^nD_{\mathcal{O}_{j_i}}} G^{(n)}_{j_1 \ldots j_n }( x_1, \ldots, x_n; \mu, g(\frac{\mu}{\lambda}))\,,
	\end{align}
	where $D_{\mathcal{O}_i}$ is the canonical dimension of $\mathcal{O}_i(x)$, and
	\begin{equation}\label{eqZ}
		\Bigg(\frac{\partial}{\partial g} + \frac{\gamma(g)}{\beta(g)}\Bigg)Z(\lambda) = 0
	\end{equation}
	in matrix notation, and
	\begin{equation} \label{ZZ}
		Z(\lambda) = P\exp \Big(\int_{g(\mu)}^{g(\frac{\mu}{\lambda})}\frac{\gamma(g')}{\beta(g')} dg'\Big)\,.
	\end{equation}
	The question arises whether a basis of renormalized operators exists where $Z(\lambda)$ becomes diagonal, so that Eq. \eqref{csformula} is greatly simplified, reducing to a single term.\par
	In essence, to address this question, we interpret \cite{Bochicchio:2021geometry} a finite change of renormalization scheme—that is, a change in the basis of renormalized operators expressed in matrix notation
	\begin{equation}\label{linearcomb}
		\mathcal{O}'(x) = S(g) \mathcal{O}(x)
	\end{equation}
	as a formal real-analytic invertible gauge transformation $S(g)$ \footnote{Obviously, in this context the gauge transformation $S(g)$ only depends on the coupling $g$ and it has nothing to do with the spacetime gauge group of the theory.} \cite{Bochicchio:2021geometry}. Under the action of $S(g)$, the matrix
	\begin{equation}
		A(g) = -\frac{\gamma(g)}{\beta(g)} = \frac{1}{g} \Big(\frac{\gamma_0}{\beta_0} + \cdots\Big)
	\end{equation}
	associated with the differential equation for $Z(\lambda)$
	\begin{equation}
		\Big(\frac{\partial}{\partial g} - A(g)\Big) Z(\lambda) = 0
	\end{equation}
	can be seen as a connection $A(g)$
	\begin{eqnarray} \label{sys2}
		A(g)= \frac{1}{g} \left(A_0 + \sum^{\infty}_ {n=1} A_{2n} g^{2n} \right)\,,
	\end{eqnarray}
	with a regular singularity at $g = 0$ that transforms as
	\begin{equation}
		A'(g) = S(g)A(g)S^{-1}(g) + \frac{\partial S(g)}{\partial g} S^{-1}(g)\,,
	\end{equation}
	with
	\begin{equation}
		\mathcal{D} = \frac{\partial }{\partial g} - A(g)
	\end{equation}
	as the corresponding covariant derivative.
	As a result, $Z(\lambda)$ can be interpreted as a Wilson line that transforms as
	\begin{equation}
		Z'(\lambda) = S(g(\mu))Z(\lambda)S^{-1}(g(\frac{\mu}{\lambda}))\,.
	\end{equation}
	\begin{theorem}\label{th:1}  \cite{Bochicchio:2021geometry}
		If the matrix $\frac{\gamma_0}{\beta_0}$ is diagonalizable and nonresonant, i.e., its eigenvalues, ordered nonincreasingly, $\lambda_1,\lambda_2,\ldots$ satisfy
		\begin{equation}
			\label{eq:nonresonant}
			\lambda_i-\lambda_j\neq 2k\ , \qquad i>j\ , \qquad k\in\mathbb{N}^+\,,
		\end{equation}
		then a formal holomorphic gauge transformation $S(g)$ exists that puts $A(g)$ into the canonical nonresonant form 
		\begin{equation} \label{1loop}
			A'(g) = \frac{\gamma_0}{\beta_0}\frac{1}{g}
		\end{equation}
		which is one-loop exact to all orders of perturbation theory. 
		Consequently, $Z(\lambda)$ is diagonalizable as well, with eigenvalues
		\begin{equation}\label{eq:oneloop}
			Z_{\mathcal{O}_i}(\lambda) = \Bigg(\frac{g(\mu)}{g(\frac{\mu}{\lambda})}\Bigg)^{\frac{\gamma_{0\mathcal{O}_i}}{\beta_0}} \,,
		\end{equation}
		where $\gamma_{0\mathcal{O}_i}$ are the eigenvalues of $\gamma_0$.
	\end{theorem}
	In this paper, we demonstrate for twist-$2$ operators in SU($N$) SYM theory that the matrix $\frac{\gamma_0}{\beta_0}$, which is already known to be diagonal \cite{Belitsky:1998gc,Braun:2003rp}, fulfills the aforementioned nonresonant condition, thereby proving the existence of the corresponding diagonal nonresonant renormalization scheme.  
	\section{Plan of the paper}
	Section \ref{sec:s1} defines the twist-$2$ operators in SU($N$) SYM theory and presents their one-loop anomalous dimensions \cite{Belitsky:2004sc}.\par
	Section \ref{sec:s2} reviews key number-theoretic concepts, such as the $p$-adic order and the classical proof demonstrating that the Harmonic numbers $H_n$ are not integers.  \par
	Section \ref{sec:s3} provides the proof of the nonresonant condition for all twist-$2$ operators in SU($N$) SYM theory.
	
	\section{Anomalous dimensions of twist-$2$ operators in SYM theory}\label{sec:s1}
	\subsection{Twist-$2$ operators in SYM theory}
	In the standard basis \cite{BPS1,Belitsky:2004sc, Belitsky:2003sh}, the gauge-invariant collinear twist-$2$ operators in the light-cone gauge that respectively appear as components of the balanced and unbalanced superfields \cite{SS1}\footnote{We refer to composite superfields made of two elementary superfields of opposite chirality as balanced, and to
	those made of two elementary superfields with the same chirality as unbalanced.}, are given by
	\begin{equation}\label{1000}
		\resizebox{0.57\textwidth}{!}{%
			$\begin{aligned}
				O^A_s &= \frac{1}{2} \partial_+ \bar{A}^a(i\overrightarrow{\partial}_++ i\overleftarrow{\partial}_+)^{s-2}C^{\frac{5}{2}}_{s-2}\Bigg(\frac{\overrightarrow{\partial}_+- \overleftarrow{\partial}_+}{\overrightarrow{\partial_+}+\overleftarrow{\partial}_+}\Bigg)\partial_+ A^a \\
				\tilde{O}^A_s &= \frac{1}{2} \partial_+ \bar{A}^a(i\overrightarrow{\partial}_++ i\overleftarrow{\partial}_+)^{s-2}C^{\frac{5}{2}}_{s-2}\Bigg(\frac{\overrightarrow{\partial}_+- \overleftarrow{\partial}_+}{\overrightarrow{\partial_+}+\overleftarrow{\partial}_+}\Bigg)\partial_+ A^a\\
				O^\lambda_s &=  \frac{1}{2} \bar{\lambda}^a(i\overrightarrow{\partial}_++ i\overleftarrow{\partial}_+)^{s-1}C^{\frac{3}{2}}_{s-1}\Bigg(\frac{\overrightarrow{\partial}_+- \overleftarrow{\partial}_+}{\overrightarrow{\partial_+}+\overleftarrow{\partial}_+}\Bigg) \lambda^a\\
				\tilde{O}^\lambda_s &=  \frac{1}{2} \bar{\lambda}^a(i\overrightarrow{\partial}_++ i\overleftarrow{\partial}_+)^{s-1}C^{\frac{3}{2}}_{s-1}\Bigg(\frac{\overrightarrow{\partial}_+- \overleftarrow{\partial}_+}{\overrightarrow{\partial_+}+\overleftarrow{\partial}_+}\Bigg) \lambda^a\\
				M_s &= \frac{1}{2}\hspace{0.08cm} \partial_+A^a (i\overrightarrow{\partial}_++ i\overleftarrow{\partial}_+)^{s-1}P^{(2,1)}_{s-1}\Bigg(\frac{\overrightarrow{\partial}_+- \overleftarrow{\partial}_+}{\overrightarrow{\partial_+}+\overleftarrow{\partial}_+}\Bigg) \lambda^a \\
				\bar{M}_s &=  \hspace{0.08cm}\frac{1}{2}\bar{\lambda}^a(i\overrightarrow{\partial}_++ i\overleftarrow{\partial}_+)^{s-1}P^{(1,2)}_{s-1}\Bigg(\frac{\overrightarrow{\partial}_+- \overleftarrow{\partial}_+}{\overrightarrow{\partial_+}+\overleftarrow{\partial}_+}\Bigg) \partial_+\bar{A}^a\, ,
			\end{aligned}$
		} 
	\end{equation}
	where $O^A_s$ and $O^\lambda_s$ are even spin operators while $\tilde{O}^A_s$ and $\tilde{O}^\lambda_s$ are odd spin operators. For the unbalanced operators
	\begin{equation}\label{1001}
		\resizebox{0.57\textwidth}{!}{%
			$\begin{aligned}
				S^A_s &= \frac{1}{2\sqrt{2}} \partial_+ \bar{A}^a(i\overrightarrow{\partial}_++ i\overleftarrow{\partial}_+)^{s-2}C^{\frac{5}{2}}_{s-2}\Bigg(\frac{\overrightarrow{\partial}_+- \overleftarrow{\partial}_+}{\overrightarrow{\partial_+}+\overleftarrow{\partial}_+}\Bigg)\partial_+ \bar{A}^a \\
				\bar{S}^A_s &= \frac{1}{2\sqrt{2}} \partial_+ A^a(i\overrightarrow{\partial}_++ i\overleftarrow{\partial}_+)^{s-2}C^{\frac{5}{2}}_{s-2}\Bigg(\frac{\overrightarrow{\partial}_+- \overleftarrow{\partial}_+}{\overrightarrow{\partial_+}+\overleftarrow{\partial}_+}\Bigg)\partial_+ A^a\\
				S^\lambda_s &=  \frac{1}{2\sqrt{2}} \bar{\lambda}^a(i\overrightarrow{\partial}_+ + i\overleftarrow{\partial}_+)^{s-1}C^{\frac{3}{2}}_{s-1}\Bigg(\frac{\overrightarrow{\partial}_+ - \overleftarrow{\partial}_+}{\overrightarrow{\partial}_++\overleftarrow{\partial}_+}\Bigg) \bar{\lambda}^a \\
				\bar{S}^\lambda_s &=  \frac{1}{2\sqrt{2}} \lambda^a(i\overrightarrow{\partial}_+ + i\overleftarrow{\partial}_+)^{s-1}C^{\frac{3}{2}}_{s-1}\Bigg(\frac{\overrightarrow{\partial}_+ - \overleftarrow{\partial}_+}{\overrightarrow{\partial}_++\overleftarrow{\partial}_+}\Bigg) \lambda^a \\
				T_s& =  \frac{1}{2}\lambda^a(i\overrightarrow{\partial}_+ + i\overleftarrow{\partial}_+)^{s-1}P^{(1,2)}_{s-1}\Bigg(\frac{\overrightarrow{\partial}_+ - \overleftarrow{\partial}_+}{\overrightarrow{\partial}_++\overleftarrow{\partial}_+}\Bigg) \partial_+\bar{A}^a \\
				\bar{T}_s &=
				\frac{1}{2} \partial_+A^a(i\overrightarrow{\partial}_+ + i\overleftarrow{\partial}_+)^{s-1}P^{(2,1)}_{s-1}\Bigg(\frac{\overrightarrow{\partial}_+ - \overleftarrow{\partial}_+}{\overrightarrow{\partial}_++\overleftarrow{\partial}_+}\Bigg) \bar{\lambda}^a\,,
			\end{aligned}$
		} 
	\end{equation}
	with $S^A_s$ and $S^\lambda_s$ being even spin operators and where $C^{\alpha}_l(x)$ are the Gegenbauer polynomials \cite{Braun:2003rp}. These operators represent the restriction to components with the maximal spin projection $s$ along the $p_+$ direction of linear combinations of twist-$2$ operators of the form
	\begin{align} \label{1}
		&O^{A\,\mathcal{T}=2}_{s} \quad=\quad \Tr\, F^\mu_{(\rho_1}\overleftarrow{D}_{\rho_2}\ldots \overrightarrow{D}_{\rho_{s-1}}F_{\rho_s)\mu}-\,\text{traces}\qquad\qquad \nonumber\\
		&\tilde{O}^{A\,\mathcal{T}=2}_{s} \quad=\quad \Tr\, \tilde{F}^\mu_{(\rho_1}\overleftarrow{D}_{\rho_2}\ldots \overrightarrow{D}_{\rho_{s-1}}F_{\rho_s)\mu}-\,\text{traces}\qquad\qquad \nonumber\\
		&O^{\lambda\,\mathcal{T}=2}_{s} \quad=\quad \Tr\, \bar{\chi}\gamma_{(\rho_1}\overleftarrow{D}_{\rho_2}\ldots \overrightarrow{D}_{\rho_{s-1})}\chi-\,\text{traces}\qquad\qquad \nonumber\\
		&\tilde{O}^{\lambda\,\mathcal{T}=2}_{s} \quad=\quad \Tr\, \bar{\chi}\gamma_{(\rho_1}\gamma_5\overleftarrow{D}_{\rho_2}\ldots \overrightarrow{D}_{\rho_{s-1})}\chi-\,\text{traces}\qquad\qquad \nonumber\\
		&M^{\mathcal{T}=2}_{s} \quad=\quad \Tr\, F_{(\rho_1}^{\nu}\overleftarrow{D}_{\rho_2}\ldots \overrightarrow{D}_{\rho_{s-1})}\sigma_{\nu}\lambda-\,\text{traces}\nonumber\\
		&\bar{M}^{\mathcal{T}=2}_{s} \quad=\quad \Tr\,\bar{\lambda}\,\bar{\sigma}_{\nu} \overleftarrow{D}_{(\rho_{s-1}}\ldots \overrightarrow{D}_{\rho_{2}} F_{\rho_{1})}^{\nu}-\,\text{traces}\nonumber\\
		&\mathbb{S}^{A\,\mathcal{T}=2}_{s}=\quad \Tr\, (F_{\mu(\nu}+i\tilde{F}_{\mu(\nu})\overleftarrow{D}_{\rho_1}\ldots \overrightarrow{D}_{\rho_{s-2}}(F_{\lambda)\sigma}+i\tilde{F}_{\lambda)\sigma})-\,\text{traces}\nonumber\\
		&S^{\lambda\,\mathcal{T}=2}_{s} \quad=\quad \Tr\, \bar{\chi}\sigma_{\mu(\rho_1}\overleftarrow{D}_{\rho_2}\ldots \overrightarrow{D}_{\rho_{s-1})}\chi-\,\text{traces}\qquad\qquad \nonumber\\
		&T^{\mathcal{T}=2}_{s+\frac{1}{2}} \quad=\quad \Tr\, F_{(\rho_1}^{\nu}\overleftarrow{D}_{\rho_2}\ldots \overrightarrow{D}_{\rho_{s-1}}\sigma_{\rho_s)\nu}\chi-\,\text{traces}\,,
	\end{align}
	including all possible combinations of right and left derivatives \cite{makeenko,Belitsky:2007jp}, where the parentheses indicate symmetrization of the enclosed indices and the trace subtraction ensures that any two-index contraction vanishes. \par
	To the leading order in perturbation theory, appropriate linear combinations of these twist-$2$ operators are conserved \cite{makeenko,Belitsky:2007jp}, and they automatically transform as primary operators under the conformal group \cite{Beisert:2004fv,makeenko,Belitsky:2007jp}.
	\subsection{Balanced and unbalanced superfields}
	For the balanced superfields we get \cite{SS1}
	\begin{align}
	\mathbb{W}_s(x,\theta,\bar{\theta}) \sim \mathbb{S}^{(2)}_{s+1}+\theta \bar{M}_{s+1}+\bar{\theta}M_{s+1}+\theta\bar{\theta}\mathbb{S}^{(1)}_{s+2}\, ,
	\end{align}
	where $\mathbb{S}^{(i)}= \{S^{(i)},\tilde{S}^{(i)}\}$ include both even- and odd-spin operators.
	For even spin with $s \geq 2$
	\begin{align}
		\label{S}
		& S^{(1)}_s = \frac{6}{s-1}O^A_s -  O^\lambda_s \nonumber\\
		& S^{(2)}_s = \frac{6}{s+2}\, O^A_s +O^\lambda_s
	\end{align}
	and for odd spin with $s\geq 3$
	\begin{align}
		\label{Stilde}
		& \tilde{S}^{(1)}_s =   -\frac{6}{s-1}\tilde{O}^A_s -  \tilde{O}^\lambda_s \nonumber\\
		& \tilde{S}^{(2)}_s =  -\frac{6}{s+2}\,\tilde{O}^A_s + \tilde{O}^\lambda_s\, ,
	\end{align}
	where $\tilde{O}^A_s$ is not defined for $s = 1$, whereas $\tilde{O}^\lambda_1$ is defined, and $\tilde{S}^{(2)}_1 =\tilde{O}^\lambda_1$.\par
	These operators also diagonalize the anomalous dimension matrix to order $g^2$, where SYM theory is conformally invariant in the conformal scheme \cite{Braun:2003rp}. \par
	Similarly, we get for the unbalanced superfields \cite{SS1}
		\begin{align}
	\mathbb{W}^+_s(x,\theta,\bar{\theta}) \sim T_{s-1}+\theta S^A_s+\bar{\theta}\bar{S}^\lambda_{s}+\theta\bar{\theta}T_{s}+\theta\bar{\theta}\,\partial_+T_{s-1}
	\end{align}
	and
	\begin{align}
	\mathbb{W}^-_s(x,\theta,\bar{\theta}) \sim \bar{T}_{s-1}+\theta \bar{S}^A_s+\bar{\theta}S^\lambda_{s}+\theta\bar{\theta}\bar{T}_{s}+\theta\bar{\theta}\,\partial_+\bar{T}_{s-1}\,.
	\end{align}
	These operators also diagonalize the anomalous dimension matrix to order $g^2$  \cite{SS1}.

	\subsection{Anomalous dimensions}\label{sec:s12}
	
	The maximal-spin components of the operators $\mathcal{O}_s$ mentioned above only mix with derivatives along the $p_+$ direction of operators of the same type but with lower spin and identical canonical dimensions \cite{Braun:2003rp,Belitsky:1998gc}. We define the bare operators for $s\geq k$ as
	\begin{equation}
		\mathcal{O}^{B\,(k)}_s =(i\partial_+)^{k}\mathcal{O}^{B}_s
	\end{equation}
	which, at the leading order of perturbation theory and for $k >0$, are conformal descendants of the corresponding primary conformal operator $\mathcal{O}^{B\,(0)}_s = \mathcal{O}^{B}_s $. 
	As a result of operator mixing, we obtain for the renormalized operators \cite{Braun:2003rp,Belitsky:1998gc}
	\begin{equation}\label{m}
		\mathcal{O}^{(k)}_s= \sum_{s\geq i \geq 2} Z_{si}\, \mathcal{O}^{B\,(k+s-i)}_{i}\,,
	\end{equation}
	where the bare mixing matrix $Z$ is lower triangular\footnote{$Z$, which in dimensional regularization depends on $g$ and $\epsilon$, should not be confused with $Z(\lambda)$.} \cite{Braun:2003rp,Belitsky:1998gc}.\par
	Therefore, the anomalous-dimension matrix $\gamma(g)$ is generally lower triangular, though $\gamma_0$ is diagonal. 
	The eigenvalues of $\gamma_0$ for $\mathcal{O}_s = S^{(1)}_s,S^{(2)}_s,\tilde{S}^{(1)}_s,\tilde{S}^{(2)}_s,M_s,\bar{M}_s$ are given by
	\cite{Belitsky:2004sc}
	\begin{equation}
		\gamma_{0\,\mathcal{O}_s} = \frac{1}{4 \pi^2} \Big( \tilde{\gamma}^0_{\mathcal{O}_s} - \frac{3}{2}\Big)\,
	\end{equation}
	with
	\begin{align}
		&\tilde{\gamma}_{0\,S^{(1)}_s} = \psi(s + 2)+\psi(s-1) - 2\psi(1) - \frac{2(-1)^s}{(s+1)s(s-1)} \nonumber\\
		&\tilde{\gamma}_{0\,S^{(2)}_s} = \psi(s + 3)+\psi(s) - 2\psi(1) + \frac{2(-1)^s}{(s+2)(s+1)s}
	\end{align}
	and ~\cite{Belitsky:2004sc,Belitsky:2003sh}
	\begin{align}\label{eq:otherbalanced}
		&\tilde{\gamma}_{0\,{\tilde{S}}^{(i)}_s}=\tilde{\gamma}_{0\,S^{(i)}_s}\nonumber\\
		&\tilde{\gamma}_{0\,M_s} =\tilde{\gamma}_{0\,S^{(2)}_s} 
	\end{align}
	Besides \cite{Belitsky:2004sc},
	\begin{equation}
		\gamma_{0\,\tilde{O}^\lambda_1} = \frac{1}{4 \pi^2}  \frac{2}{3} = \frac{1 }{6 \pi^2} \,.
	\end{equation}
	For the operators $\mathcal{O} = S^{A},\bar{S}^{A},S^{\lambda},\bar{S}^{\lambda},T,\bar{T}$, $\gamma_0$ is diagonal, with eigenvalues \cite{Belitsky:2004sc}
	\begin{equation}
		\gamma_{0\mathcal{O}_s} = \frac{1}{4 \pi^2} \Big( \tilde{\gamma}_{0\mathcal{O}_s} - \frac{3}{2}\Big)\,,
	\end{equation}
	where \cite{Belitsky:2004sc}
	\begin{align}
		&\tilde{\gamma}_{0S^{A}_s} = 2\psi(s + 1)- 2\psi(1) \nonumber\\
		&\tilde{\gamma}_{0S^{\lambda}_s} = 2\psi(s + 1)- 2\psi(1) 
	\end{align}
	and \cite{Belitsky:2004sc,Belitsky:2003sh}
	\begin{align}
		\tilde{\gamma}_{0T_s} 
		=\begin{cases}
			\psi(s+1)-\psi(1),&s=2, 4 , \ldots\\
			\psi(s+2)-\psi(1), & s=1,3, \ldots \,.
		\end{cases}
	\end{align}
	We have numerically verified that the initial $10^4$ eigenvalues of $\frac{{\gamma}_{0\mathcal{O}_s}}{\beta_0}$ are nonresonant, with $\beta_0 = \frac{3}{(4\pi)^2}$.\par
	The eigenvalues of $\gamma_0$ are naturally ordered in increasing sequence with increasing $s$, which is contrary to the ordering in Theorem \ref{th:1}. However, it can be easily seen from the proof in \cite{Bochicchio:2021geometry} that for this case, the nonresonant condition takes the form
	\begin{equation}
		\label{eq:nonresonant2}
		\lambda_j-\lambda_i\neq 2k\ , \qquad j>i\ , \qquad k\in\mathbb{N}^+\,,
	\end{equation}
	with $\lambda_1\leq \lambda_2\leq \lambda_3\leq \ldots$.
	\section{Number-theoretic concepts}\label{sec:s2}
	\subsection{$p$-adic order}
	The $p$-adic order of an integer $n$ is defined as the exponent of the highest power of a prime number $p$ that divides $n$ \cite{ireland1990classical}. 
	Specifically, the $p$-adic order of an integer is the function
	\begin{equation}
		\nu_p(n)=
		\begin{cases}
			\mathrm{max}\{k \in \mathbb{N} : p^k\,\text{divides}\,n\} & \text{if } n \neq 0\\
			\infty & \text{if } n=0\,.
		\end{cases}
	\end{equation}
	For example, $\nu_3(24) = \nu_3(3\times2^3) = 1$ and $\nu_2(24) = 3$.\par
	This concept can be extended to rational numbers through the property \cite{ireland1990classical}
	\begin{equation}
		\label{eq:p1}
		\nu_p\left(\frac{a}{b}\right)  =\nu_p(a)-\nu_p(b)\,.
	\end{equation}
	Consequently, rational numbers may have a negative $p$-adic order, whereas integers can only have non-negative values for any prime $p$.
	Additional properties include \cite{ireland1990classical}
	\begin{align}
		\label{eq:p2}
		&\nu_p(a\cdot b) = \nu_p(a)+\nu_p(b)\nonumber\\
		&\nu_p(a+b)\geq\min\bigl\{ \nu_p(a), \nu_p(b)\bigr\}\,.
	\end{align}
	When $\nu_p(a) \neq \nu_p(b)$ 
	\begin{equation}
		\label{equal}
		\nu_p(a+b)= \min\bigl\{ \nu_p(a), \nu_p(b)\bigr\}
	\end{equation}
	\cite{ireland1990classical}, a fact which is crucial for the subsequent proof.
	\subsection{Harmonic numbers and Bertrand's postulate}
	Bertrand's postulate\footnote{It is actually a theorem.} \cite{cebby,ramanujan1919proof} asserts that for any real number $x\geq2$, there is at least one prime number $p$ satisfying
	\begin{equation}\label{eq:bertrand}
		\frac{x}{2}+1\leq p\leq x\,.
	\end{equation} 
	This implies that for any prime $p\in\left[\frac{x}{2}+1,x\right]$, its double, $2p$, cannot lie within the same interval, since $2p\geq x+2$. \par
	We will apply Bertrand's postulate to prove the classical result that Harmonic numbers, $H_n$,
	\begin{align}\label{eq:Hn}
		H_n = \sum_{k = 1}^{n}\frac{1}{k}\,,
	\end{align}
	are never integers for any $n\geq 2$.
	\subsection{Standard argument}
	Let $p$ be a prime number within the interval 
	\begin{equation}\label{eq:postulaten}
		\frac{n}{2}+1\leq p\leq n\,.
	\end{equation}
	For such a prime $p$, the term $\frac{1}{p}$ appears in the summation of Eq. \eqref{eq:Hn}. However, no term with $k>p$ can have $p$ as a prime factor, as its prime factorization would have to include at least $2p$, which falls outside the specified interval. Thus, with the exception of $\frac{1}{p}$, the denominator $k$ of every term $\frac{1}{k}$ in the sum is divisible only by primes other than $p$. Consequently, if we write the sum as
	\begin{equation}
		\sum_{k=1}^n\frac{1}{k} = \frac{1}{p} + \frac{a}{b}\,,
	\end{equation}
	then the denominator $b$ is not divisible by $p$, i.e., $\gcd(b,p)=1$. This leads to the conclusion that Harmonic numbers are not integers. To be more precise,
	\begin{equation}
		\nu_p\left(\frac{1}{p}\right)=-1\qquad \text{while}\qquad \nu_p\left(\frac{a}{b}\right) = \nu_p(a)>0\,.
	\end{equation} 
	Then, observing that $\nu_p\left(\frac{1}{p}\right)\neq \nu_p\left(\frac{a}{b}\right)$, it follows from Eq. \eqref{equal}
	\begin{equation}
		\nu_p\left(\frac{1}{p}+\frac{a}{b}\right) =\min\left(\nu_p\left(\frac{1}{p}\right),\nu_p\left(\frac{a}{b}\right)\right) = -1
	\end{equation}
	and, finally,
	\begin{equation}
		\nu_p(H_n) = \nu_p\left(\frac{1}{p} + \frac{a}{b}\right) = \min\left(\nu_p\left(\frac{1}{p}\right),\nu_p\left(\frac{a}{b}\right)\right) = -1\,.
	\end{equation}
	Thus, because the $p$-adic order of $H_n$ is negative, $H_n$ cannot be an integer. \\
	This line of reasoning will be referred to as the \textit{standard argument}, since it will be used multiple times in the subsequent sections. 
	\subsection{Generalized argument}\label{sec:generalizedsec}
	We now consider sums that are more complex than the Harmonic numbers
	\begin{equation}\label{eq:general}
		\Omega_n = \sum_{k=1}^{n}\frac{c_k}{k}\,.
	\end{equation}
	We start with the case where the coefficients $c_k$ can assume positive or negative values, such as $\pm1$ or $\pm2$.\par
	The coefficients $\pm1$ clearly do not affect the standard argument, as $1$ is coprime with any prime $p$. If $p$ is found by Bertrand's postulate, there is no $k\neq p$ in the sum that has $p$ in its prime factorization. Therefore, $\Omega_n$ is not an integer according to the standard argument. \par 
	Slightly more caution is required when $\pm2$ appears in the numerators. For $n\geq3$, Bertrand's postulate guarantees the existence of a prime $ p\geq 3$ satisfying Eq. \eqref{eq:postulaten} such that 
	\begin{equation}
		\Omega_n  = \frac{c_p}{p}+\frac{a}{b}
	\end{equation}
	where $\gcd(b,p)=\gcd(c_p,p)=1$. In this scenario, $c_p$ and $p$ are indeed coprime and, as before, no other denominator $k$ in the sum shares $p$ as a prime factor, even if other coefficients $c_k$ take values $\pm1, \pm2$. \par 
	For instance, let us suppose there are terms with denominators
	\begin{equation}\label{eq:pp}
		k_1=p'
	\end{equation}
	and
	\begin{equation}\label{eq:2pp}
		k_2=2p'
	\end{equation}
	such that $c_{k_2}=\pm2$; then the terms $\frac{c_{k_1}}{k_1}=\frac{c_{k_1}}{p'}$ and $\frac{\pm 2}{k_2}=\frac{\pm1}{p'}$ would combine to form $\frac{c_{p'}}{p'}$, where $c_{p'}$ could potentially be zero. However, according to Eqs. \eqref{eq:pp} and \eqref{eq:2pp}, $p'$ is not one of the primes identified by Bertrand's postulate. Therefore, its potential absence from the sum does not impact the standard argument, which remains valid.\par
	More broadly, if $c_k\in \mathbb{Z} \setminus \{0\}$ and if a prime $p$ can be found that satisfies Eq. \eqref{eq:postulaten} and $\gcd(c_p,p)=1$, then the standard argument remains applicable. This is because, according to Bertrand's postulate, all terms in the sum other than $\frac{c_p}{p}$ will combine into a fraction whose denominator is not divisible by $p$.  It is clear by the nature of the sum in Eq. \eqref{eq:general} that there is an implicit condition on $n$ that must be satisfied, namely, given a certain $c_k\in \mathbb{Z} \setminus \{0\}$, the value of $n$ must be such that a suitable $p$ exists for the requirement $\gcd(c_p,p)=1$ to be fulfilled. This condition on $n$ is something that must be checked on a case by case situation for every proof given the sequence $c_k$.\par 
	Finally, we examine the most complex case, where some coefficients $c_k$ are permitted to be zero. For the argument to hold, there must be at least one prime $p$ satisfying Eq. \eqref{eq:postulaten} for which $c_p\neq 0$ and $\gcd(c_p,p)=1$. When this condition is met, even with an arbitrary number of zero coefficients $c_k$, the standard argument can still be applied to demonstrate that $K_n$ is not an integer.\par 
	Evidently, verifying this condition requires a direct, case-by-case inspection of the sum.
	
	\section{Proof of the nonresonant condition}\label{sec:s3}
	This section provides the proof that the eigenvalues of the anomalous dimensions for the aforementioned twist-$2$ operators are nonresonant, as defined in Eq. \eqref{eq:nonresonant2}. \par
	We begin by recalling that the digamma function can be expressed as
	\begin{equation}
		\label{gammaH}
		\psi(n+1) = H_n-\gamma\,,
	\end{equation}
	where $\gamma$ is the Euler-Mascheroni constant. 
	\subsection{Nonresonant condition for unbalanced twist-$2$ operators}\label{sec:proof1}
	Using Eq. \eqref{gammaH}, we express the anomalous dimension of $S^A_s$ and $S^\lambda_s$ in a more suitable form
	\begin{align}
		\gamma^{S^A}_{0n}=	\gamma^{S^\lambda}_{0n}=\frac{2}{(4 \pi)^2} \left(2H_n-\frac{3}{2}\right)\,,
	\end{align}
	with $n=2,4,6,\ldots$.
	\begin{lemma}\label{lemma:1}
		The sequence $\gamma^{S^A}_{0n}$ is monotonically increasing
		\begin{equation}
			\gamma^{S^A}_{0n+1}\geq \gamma^{S^A}_{0n}
		\end{equation}
	\end{lemma}
	\textit{Proof.}\\
	\begin{equation}
		\gamma^{S^A}_{0n+1}-\gamma^{S^A}_{0n} = \frac{4}{(4\pi)^2}\frac{1}{n+1}> 0 \,.
	\end{equation}
	\qed\\\\
	Therefore, the sequence $\gamma_{0n}^{S^A}$ (and $\gamma_{0n}^{S^\lambda}$) is increasing and matches the ordering in Eq. \eqref{eq:nonresonant2}.\\
	\begin{theorem}
		The eigenvalues of $\frac{\gamma^{S^{A,\lambda}}_0}{\beta_0}$ are nonresonant
		\begin{equation}\label{eq:th1}
			\frac{\gamma_{0n}^{S^{A,\lambda}}-\gamma_{0m}^{S^{A,\lambda}}}{\beta_0} \neq 2k\,,\qquad k\in\mathbb{N}^+,\quad \forall n>m\geq2\,,
		\end{equation}
		where $\beta_0 = \frac{3}{(4\pi)^2}$.\\ 
	\end{theorem}
	\textit{Proof.}\\
	Let us set $n =m+x$, with $x\geq1$ being a natural number. Equation \eqref{eq:th1} can then be written as
	\begin{equation}
		\Delta_m^{S^{A,\lambda}} (x)=\frac{\gamma_{0{m+x}}^{S^{A,\lambda}}-\gamma_{0m}^{S^{A,\lambda}}}{\beta_0} = \frac{4}{3}\sum_{k = m+1}^{m+x}\frac{1}{k} = \frac{4}{3}\Sigma_m(x)\,,
	\end{equation}
	with
	\begin{equation}
		\Sigma_m(x) = \sum_{k = m+1}^{m+x}\frac{1}{k}\,.
	\end{equation}
	We parametrize $x$ as
	\begin{align}
		x = m+t\qquad t\geq0\,,
	\end{align}
	which leaves out the first part of the proof for all possible values of $x<m$; we will consider these later. Hence,
	\begin{align}
		\Delta_m^{S^{A,\lambda}} (m+t)= \frac{4}{3}\sum_{k = m+1}^{2m+t}\frac{1}{k}\,.
	\end{align}
	By Bertrand's postulate, there again exists a prime $p$ in the interval 
	\begin{equation}\label{eq:pin1}
		m+1+\frac{t}{2}\leq p\leq 2m+t\,.
	\end{equation}
	Therefore, applying the standard argument yields 
	\begin{equation}
		\nu_p(\Sigma_m(m+t)) =-1\,.
	\end{equation}
	Thus, using Eqs. \eqref{eq:p1} and \eqref{eq:p2}
	\begin{align}
		\nu_p\left(\frac{4}{3}	\Sigma_m(m+t)\right) &= \nu_p\left(\frac{4}{3}\right)+\nu_p(	\Sigma_m(m+t)) \nonumber\\
		&= \nu_p(4)-\nu_p(3)-1\,.
	\end{align} 
	Then, for $m>2$ or for $m=2$ with $t>0$, Eq. \eqref{eq:pin1} implies $p\geq 5$, so that $\nu_p(4)=0$ and
	\begin{equation}
		\nu_p\left(\frac{4}{3}	\Sigma_m(m+t)\right) <0.
	\end{equation}
	For the special case $m=2$ and $t= 0$, we find $\Sigma_2(2) = \frac{7}{12}$ and  $	\Delta_2^{S^{A,\lambda}} (2) = \frac{7}{9}$. We thus conclude that for $x\geq m$, the $p$-adic order of $\nu_p(\Delta_m^{S^{A,\lambda}}(x))<0$, which means it cannot be an integer.\par
	The preceding part of the proof did not consider the case where $x<m$. We now address this case, beginning by demonstrating as in Proposition \ref{prop:Sigma-main} that
	\begin{align}
\label{eq:ineq_rephrased}
		\Sigma_m(x) <\log(2)<1\qquad\qquad\forall x<m\,.
	\end{align}	
	Hence, Eq. \eqref{eq:ineq_rephrased} implies
	\begin{align}
		\Delta_m^{S^{A,\lambda}}(x)= 	\frac{4}{3}\Sigma_m(x) <\frac{4}{3}\log(2)<0.93\qquad \forall x<m\, .
	\end{align}
	Therefore, $\Delta_m^{S^{A,\lambda}}(x)$ is strictly less than $1$.\\
	In conclusion, $\Delta_m^{S^{A,\lambda}}$ cannot be an integer for $x\geq m$, nor can it be an integer for $x<m$. This completes the proof of the theorem.
	\qed\\\\
	Using Eq. \eqref{gammaH}, we now write the anomalous dimension of $T_s$ in a more suitable form
	\begin{align}
		\gamma^{T}_{0n}=\begin{cases}
			\frac{2}{(4 \pi)^2} \left(H_n-\frac{3}{2}\right),&n=2, 4 , \ldots\\
			\frac{2}{(4 \pi)^2} \left(H_{n+1}-\frac{3}{2}\right), & n=1,3, \ldots \,.
		\end{cases}
	\end{align}
	Since the even and odd anomalous dimensions share the same functional dependence, we can parametrize them as $n=2l$ and $n=2l-1$ for even and odd spins respectively to obtain the same functional form. We can therefore work directly with
	\begin{align}
		\gamma^{T}_{0n}=
		\frac{2}{(4 \pi)^2} \left(H_{2n}-\frac{3}{2}\right),\qquad n=1,2,3,\ldots\,,
	\end{align}
	for all values of $n$.
	\begin{lemma}\label{lemma:11}
		The sequence $\gamma^{T}_{0n}$ is monotonically increasing
		\begin{equation}
			\gamma^{T}_{0n+1}\geq \gamma^{T}_{0n}
		\end{equation}
	\end{lemma}
	\textit{Proof.}\\
	\begin{equation}
		\gamma^{T}_{0n+1}-\gamma^{T}_{0n} = \frac{2}{(4\pi)^2}\left(\frac{1}{2n+1}+\frac{1}{2n+2}\right)> 0 \,.
	\end{equation}
	\qed\\\\
	Therefore, the sequence $\gamma^{T}_{0n}$ is increasing and matches the ordering in Eq. \eqref{eq:nonresonant2}.\\
	\begin{theorem}
		The eigenvalues of $\frac{\gamma^{T}_0}{\beta_0}$ are nonresonant
		\begin{equation}\label{eq:th11}
			\frac{\gamma_{0n}^{T}-\gamma_{0m}^{T}}{\beta_0} \neq 2k\,,\qquad k\in\mathbb{N}^+,\quad \forall n>m\geq2\,,
		\end{equation}
		where $\beta_0 = \frac{3}{(4\pi)^2}$.\\ 
	\end{theorem}
	\textit{Proof.}\\
	Let us set $n=m+x$, where $x\geq1$ is a natural number. Equation \eqref{eq:th11} can then be written as
	\begin{equation}
		\Delta_m^{T} (x)=\frac{\gamma_{0{m+x}}^{T}-\gamma_{0m}^{T}}{\beta_0} = \frac{2}{3}\sum_{k = 2m+1}^{2m+2x}\frac{1}{k} = \frac{2}{3}\Sigma'_m(x)\,,
	\end{equation}
	with
	\begin{equation}
		\Sigma'_m(x) = \sum_{k = 2m+1}^{2m+2x}\frac{1}{k}\,.
	\end{equation}
	We parametrize $x$ as
	\begin{align}
		x = m+t\qquad t\geq0\,,
	\end{align}
	which excludes values of $x<m$ from this part of the proof. Thus,
	\begin{align}
		\Delta_m^{T} (m+t)= \frac{2}{3}\sum_{k = 2m+1}^{4m+2t}\frac{1}{k}\,.
	\end{align}
	By Bertrand's postulate, a prime exists in the interval 
	\begin{equation}\label{eq:pin2}
		2m+1+t\leq p\leq 4m+2t\,.
	\end{equation}
	Applying the standard argument gives
	\begin{equation}
		\nu_p(\Sigma'_m(m+t)) =-1\,.
	\end{equation}
	Using Eqs. \eqref{eq:p1} and \eqref{eq:p2}, we find
	\begin{align}
		\nu_p\left(\frac{2}{3}	\Sigma'_m(m+t)\right) &= \nu_p\left(\frac{2}{3}\right)+\nu_p(	\Sigma'_m(m+t)) \nonumber\\
		&= \nu_p(2)-\nu_p(3)-1\,.
	\end{align} 
	For $m>1$, or for $m=1$ with $t>0$, Eq. \eqref{eq:pin2} implies $p\geq 5$, meaning $\nu_p(2)=\nu_p(3)=0$ and
	\begin{equation}
		\nu_p\left(\frac{2}{3}	\Sigma'_m(m+t)\right) <0.
	\end{equation}
	For the case $m=1$ and $t= 0$, $\Sigma'_1(1) = \frac{7}{12}$ and  $	\Delta_1^{T} (1) = \frac{7}{18}$. We conclude that for $x\geq m$, the $p$-adic order of $\nu_p(\Delta_m^{T}(x))$ is negative, so it cannot be an integer.\par
	We now consider the case where $x<m$. By an argument identical to the one in Proposition \ref{prop:Sigma-main} we see that
	\begin{align}
\label{eq:ineq22}
		\Sigma'_m(x) <\log(2)<1\qquad\qquad\forall x<m\,.
	\end{align}	
	Equation \eqref{eq:ineq22} therefore implies
	\begin{align}
		\Delta_m^{T}(x)= 	\frac{2}{3}\Sigma'_m(x) <\frac{2}{3}\log(2)<0.5\qquad \forall x<m\, .
	\end{align}
	Thus, $\Delta_m^{T}(x)$ is strictly less than $1$.\\
	We conclude that $\Delta_m^{T}$ cannot be an integer for $x\geq m$, nor for $x<m$. The theorem is therefore proved.
	\qed
	\subsection{Nonresonant condition for balanced twist-$2$ operators of even spin}
	We now examine the anomalous dimension of balanced operators $S^{(1)}_s$ with even spin
	\begin{equation}
		\gamma^{S^{(1)}}_{0n}= \frac{2}{(4 \pi)^2} \left(
		2 H_{n-2}+\frac{3}{n}-\frac{3}{2} \right)
	\end{equation}
	for $n=2,4,6,\ldots$.
	\begin{lemma}\label{lemma2}
		The sequence $\gamma^{S^{(1)}}_{0n}$ is monotonically increasing
		\begin{equation}
			\gamma^{S^{(1)}}_{0n+1}\geq 	\gamma^{S^{(1)}}_{0n}
		\end{equation}
	\end{lemma}
	\textit{Proof.}\\
	The difference between consecutive eigenvalues can be written as
	\begin{align}
		\gamma^{S^{(1)}}_{0n+1}-\gamma^{S^{(1)}}_{0n}	& = \frac{2}{(4\pi)^2}\frac{2n(n+1) - (n^2+n-3)}{n(n^2-1)}>0\,.
	\end{align}
	Therefore, $\gamma^{S^{(1)}}_{0n}$ is monotonically increasing and aligns with the ordering in Eq. \eqref{eq:nonresonant2}.
	\qed\\
	
	\begin{theorem}\label{th:2}
		The eigenvalues of $\frac{\gamma^{S^{(1)}}_0}{\beta_0}$ are nonresonant
		\begin{equation}
			\frac{\gamma_{0n}^{S^{(1)}}-\gamma_{0m}^{S^{(1)}}}{\beta_0} \neq 2k\,,\qquad k\in\mathbb{N}^+,\quad \forall n>m\geq 2\,,
		\end{equation}
		where $\beta_0 = \frac{3}{(4\pi)^2}$.
	\end{theorem} 
	\textit{Proof.}\\ 
	The proof follows a similar structure to that for unbalanced twist-$2$ operators, though with additional care as outlined in section \ref{sec:generalizedsec}.\par
	We again set $n =m+x$, where $x>0$ is a natural number. The difference of eigenvalues can then be expressed as
	\begin{align}
		\Delta^{S^{(1)}} _m(x)&=\frac{\gamma_{0{m+x}}^{S^{(1)}}-\gamma_{0m}^{S^{(1)}}}{\beta_0}\nonumber\\
		&= \frac{2}{3}\Bigg(2\sum_{k = m-1}^{m-2+x}\frac{1}{k}+\frac{3}{m+x}-\frac{3}{m}\Bigg)\nonumber\\
		&=\frac{2}{3}K_m(x)
	\end{align}
	with
	{\small
		\begin{align}\label{eq:si}
			K_m(x) =2\sum_{k = m-1}^{m-2+x}\frac{1}{k}+\frac{3}{m+x}-\frac{3}{m}\,.
		\end{align}
	}
	This sum clearly matches the form of Eq. \eqref{eq:general}, with coefficients $c_k=\pm2,\pm3$ and no gaps.  
	As before, we set
	\begin{align}
		x = m+t\qquad t\geq0\,,
	\end{align}
	which excludes a finite number of values $x<m$. For $x\geq m$, we have
	\begin{align}
		\Delta_m^{S^{(1)}} (m+t)&=\frac{2}{3}\left(2\sum_{k = m-1}^{2m-2+t}\frac{1}{k}+\frac{3}{2m+t}-\frac{3}{m}\right)\,.
	\end{align}
	The generalized argument applies directly for a prime in the interval 
	\begin{equation}
		m+\frac{t}{2}\leq p\leq 2m-2+t\,,
	\end{equation}
	when $m+\frac{t}{2}>3$. For $m\leq 3-\frac{t}{2}$, we must check a small number of cases directly. Notably, for $m=2$ and $t=1$, we find $\Delta_2^{S^{(1)}} (3)=1$. In all other instances, $\Delta_m^{S^{(1)}} (m+t)$ is not an integer. We conclude that $\Delta_m^{S^{(1)}} (m+t)$ is never an integer greater than $1$. \par
	We now consider the values of $x<m$.
	In this case, similar to section \ref{sec:proof1}, we show in Proposition \ref{prop:K} that the bound below holds
	\begin{align}
\label{eq:ineq2}
		K_m(x)<2\log(2)\,,\qquad \qquad\forall x<m\,.
	\end{align}	
	From Eq. \eqref{eq:ineq2}, it follows that
	\begin{align}
		\Delta_m^{S^{(1)}}(x)= 	\frac{2}{3}K_m(x) <\frac{4}{3}\log(2)<0.93\,\qquad\forall x<m\,.
	\end{align}
	We conclude that $\Delta_m^{S^{(1)}}$ cannot be an integer greater than $1$ for $x\geq m$ and cannot be an integer for $x<m$, which proves the theorem.
	\qed\\\\
	Next, we study the anomalous dimension of balanced operators $S^{(2)}_s$ of even spin
	\begin{equation}
		\gamma^{S^{(2)}}_{0n}= \frac{2}{(4 \pi)^2} \left(2 H_{n}+\frac{2(-1)^n}{(n+2)(n+1)n}-\frac{3}{2}\right)
	\end{equation}
	with $n=2,4,6,\ldots$.
	\begin{lemma}\label{lemma24}
		The sequence $\gamma^{S^{(2)}}_{0n}$ is monotonically increasing
		\begin{equation}
			\gamma^{S^{(2)}}_{0n+1}\geq 	\gamma^{S^{(2)}}_{0n}
		\end{equation}
	\end{lemma}
	\textit{Proof.}\\
	The difference of consecutive eigenvalues can be written as
	\begin{align}
		\gamma^{S^{(2)}}_{0n+1}-\gamma^{S^{(2)}}_{0n}	& = \frac{2}{(4\pi)^2}\left(\frac{2}{n+1}+\frac{2(-1)^{n+1}}{(n+3)(n+2)(n+1)}-\frac{2(-1)^n}{(n+2)(n+1)n}\right)>0\,.
	\end{align}
	Therefore, $\gamma^{S^{(2)}}_{0n}$ increases monotonically and matches the ordering in Eq. \eqref{eq:nonresonant2}.
	\qed\\
	
	\begin{theorem}\label{th:24}
		The eigenvalues of $\frac{\gamma^{S^{(2)}}_0}{\beta_0}$ are nonresonant
		\begin{equation}
			\frac{\gamma_{0n}^{S^{(2)}}-\gamma_{0m}^{S^{(2)}}}{\beta_0} \neq 2k\,,\qquad k\in\mathbb{N}^+,\quad \forall n>m\geq 2\,,
		\end{equation}
		where $\beta_0 = \frac{3}{(4\pi)^2}$.
	\end{theorem} 
	\textit{Proof.}\\ 
	The proof is analogous to the previous cases. We set $n=m+x$ for a natural number $x>0$. The difference of eigenvalues is
	\begin{align}
		\Delta^{S^{(2)}} _m(x)&=\frac{\gamma_{0{m+x}}^{S^{(2)}}-\gamma_{0m}^{S^{(2)}}}{\beta_0}\nonumber\\
		&= \frac{2}{3}\left(2\sum_{k = m+1}^{m+x}\frac{1}{k}+\frac{2(-1)^{m+x}}{(m+x+2)(m+x+1)(m+x)}-\frac{2(-1)^m}{(m+2)(m+1)m}\right)\nonumber\\
		&=\frac{2}{3}J_m(x)
	\end{align}
	with
	{\small
		\begin{align}\label{eq:si4}
			J_m(x) =2\sum_{k = m+1}^{m+x}\frac{1}{k}+\frac{2(-1)^{m+x}}{(m+x+2)(m+x+1)(m+x)}-\frac{2(-1)^m}{(m+2)(m+1)m}\,.
		\end{align}
	}
	This sum is of the form in Eq. \eqref{eq:general}. Setting $x=m+t$ for $t\geq0$, we find
	\begin{align}
		\Delta_m^{S^{(2)}} (m+t)&=\frac{2}{3}\left(2\sum_{k=m+1}^{2m+t}\frac{1}{k}+\frac{2(-1)^{2m+t}}{(2m+t+2)(2m+t+1)(2m+t)}-\frac{2(-1)^m}{(m+2)(m+1)m}\right)\,.
	\end{align}
	The generalized argument applies to a prime in the interval $m+1+\frac{t}{2}\leq p\leq 2m+t$. Since the denominators of the fractional terms are polynomials in $m$ and $t$, they will be coprime with $p$ for sufficiently large $m$. It can be verified that $\Delta_m^{S^{(2)}} (m+t)$ is never an integer.\par
	For the remaining values $x<m$, we  demonstrate in Proposition \ref{prop:J-upper} the bound
	\begin{align}
\label{eq:ineq24}
		J_m(x)<2\log(2)\,,\qquad \qquad\forall x<m\,.
	\end{align}	
	Equation \eqref{eq:ineq24} then implies
	\begin{align}
		\Delta_m^{S^{(2)}}(x)= 	\frac{2}{3}J_m(x) <\frac{4}{3}\log(2)<0.93\,\qquad\forall x<m\,.
	\end{align}
	Therefore, $\Delta_m^{S^{(2)}}$ cannot be an integer for $x\geq m$ or for $x<m$, which proves the theorem.
	\qed\\\\
	From Eq. \eqref{eq:otherbalanced}, we conclude that all other anomalous dimensions of balanced operators of even spin also satisfy the nonresonant condition.
	
	\subsection{Nonresonant condition for balanced twist-$2$ operators of odd spin}
	We now address the anomalous dimension of odd-spin balanced operators. First, for $S^{(1)}_s$,
	{\small
		\begin{align}
			\gamma^{S^{(1)}}_{0n}
			= \frac{2}{(4 \pi)^2}  \left(2 H_{n-2}+\frac{2}{n-1}-\frac{1}{n}+\frac{2}{n+1}-\frac{3}{2}  \right)
		\end{align}
	}
	with $n=3,5,7,\ldots$.
	\begin{lemma}\label{lemma3}
		The sequence $\gamma^{S^{(1)}}_{0n}$ is monotonically increasing
		\begin{equation}
			\gamma^{S^{(1)}}_{0n+1}\geq \gamma^{S^{(1)}}_{0n}
		\end{equation}
	\end{lemma}
	\textit{Proof.}\\
	The explicit difference is
	\begin{align}
		\gamma^{S^{(1)}}_{0n+1}- \gamma^{S^{(1)}}_{0n}
		& = \frac{2}{(4\pi)^2}\left(\frac{2}{n-1}+\frac{2}{n}-\frac{2}{n+1}-\frac{2}{n+2}\right)>0\,.
	\end{align}
	Thus, $\gamma^{S^{(1)}}_{0n}$ is a monotonically increasing sequence, matching the ordering in Eq. \eqref{eq:nonresonant2}.\\
	\qed\\\\
	
	\begin{theorem}
		The eigenvalues of $\frac{\gamma^{S^{(1)}}_0}{\beta_0}$ are nonresonant
		\begin{equation}
			\frac{\gamma_{0n}^{S^{(1)}}-\gamma_{0m}^{S^{(1)}}}{\beta_0} \neq 2k\,,\qquad k\in\mathbb{N}^+,\quad \forall n>m\geq 3\,,
		\end{equation}
		where $\beta_0 = \frac{3}{(4\pi)^2}$.
	\end{theorem} 
	\textit{Proof.}\\
	Following the established procedure, for $n=m+x$ with $x>0$ a natural number, the difference of eigenvalues is
	\begin{align}
		\Delta_m^{S^{(1)}}(x)&=\frac{2}{3}\left(2\sum_{k = m-1}^{m+x-2}\frac{1}{k}+\frac{2(-1)^{m+x-1}}{(m+x)(m+x-1)}-\frac{2(-1)^{m-1}}{m(m-1)}\right)
		=\frac{2}{3}	L_m(x)\,.
	\end{align}
	Setting $x=m+t$ for $t\geq0$ leaves a finite number of cases $x<m$. The generalized argument applies to the case $x \geq m$, confirming non-integrality for large enough $m$. A direct check handles the few remaining small $m$ cases.\par
	For $x<m$, 	we show in Proposition \ref{prop:L} that the following bound holds
	\begin{align}
		\label{emtilde_rephrased}
		L_m(x) <2\log(2)+\frac{3}{2},\qquad\qquad \forall x<m\,.
	\end{align}	
	This implies that
	\begin{align}
		\Delta_m^{S^{(1)}}(x)= 	\frac{2}{3}L_m(x) <\frac{4}{3}\log(2)+1<1.93\,,\qquad\qquad\forall x<m\,.
	\end{align}
	We conclude that $\Delta_m^{S^{(1)}}$ cannot be an integer greater then $1$ and in particular it cannot be an even integer, thus proving the nonresonance condition.
	\qed\\\\
	Finally, for the operators $S^{(2)}_s$ of odd spin,
	{\small
		\begin{align}
			\gamma^{S^{(2)}}_{0n}
			= \frac{2}{(4 \pi)^2}  \left(2H_{n+1}-\frac{2(-1)^n}{(n+2)(n+1)}-\frac{3}{2}\right)
		\end{align}
	}
	with $n=1,3,5,\ldots$.
	\begin{lemma}\label{lemma33}
		The sequence $\gamma^{S^{(2)}}_{0n}$ is monotonically increasing
		\begin{equation}
			\gamma^{S^{(2)}}_{0n+1}\geq \gamma^{S^{(2)}}_{0n}
		\end{equation}
	\end{lemma}
	\textit{Proof.}\\
	The difference is
	\begin{align}
		\gamma^{S^{(2)}}_{0n+1}- \gamma^{S^{(2)}}_{0n}
		& = \frac{2}{(4\pi)^2}\left(\frac{2}{n+1}+\frac{2}{n+2} -\frac{2(-1)^{n+1}}{(n+3)(n+2)} + \frac{2(-1)^{n}}{(n+1)n}\right)>0\,.
	\end{align}
	Thus, $\gamma^{S^{(2)}}_{0n}$ is monotonically increasing and aligns with the ordering in Eq. \eqref{eq:nonresonant2}.\\
	\qed\\\\
	
	\begin{theorem}
		The eigenvalues of $\frac{\gamma^{S^{(2)}}_0}{\beta_0}$ are nonresonant
		\begin{equation}
			\frac{\gamma_{0n}^{S^{(2)}}-\gamma_{0m}^{S^{(2)}}}{\beta_0} \neq 2k\,,\qquad k\in\mathbb{N}^+,\quad \forall n>m\geq 1\,,
		\end{equation}
		where $\beta_0 = \frac{3}{(4\pi)^2}$.
	\end{theorem} 
	\textit{Proof.}\\
	The proof strategy remains the same. The difference of eigenvalues
	\begin{align}
		\Delta_m^{S^{(2)}}(x)&=\frac{2}{3}\left(2\sum_{k = m+2}^{m+x+1}\frac{1}{k}-\frac{2(-1)^{m+x}}{(m+x+2)(m+x+1)} + \frac{2(-1)^m}{(m+2)(m+1)}\right)\nonumber\\
		&=\frac{2}{3}U_m(x)\,.
	\end{align}
	is not an integer for $x \geq m$ by the generalized argument. For the remaining cases $x<m$, we show in Proposition \ref{prop:U} the bound holds
	\begin{align}
		\label{emtilde3_rephrased}
		U_m(x) <2\log(2),\qquad\qquad \forall x<m\,.
	\end{align}	
	This implies
	\begin{align}
		\Delta_m^{S^{(2)}}(x)= 	\frac{2}{3}U_m(x) <\frac{4}{3}\log(2)<0.93\,,\qquad\qquad\forall x<m\,.
	\end{align}
	We conclude that $\Delta_m^{S^{(2)}}$ cannot be an integer, which completes the proof.
	\qed\\\\
	As in the even-spin case, Eq. \eqref{eq:otherbalanced} implies that all other anomalous dimensions of balanced operators of odd spin satisfy the nonresonant condition.
	\section{Conclusions}
	We have shown that the eigenvalues of the (diagonal) matrices $\frac{\gamma_0}{\beta_0}$ for the twist-$2$ operators in SUSY $\mathcal{N}=1$ SU($N$) Yang-Mills theory satisfy the nonresonant condition in \textbf{Theorem \ref{th:1}}. Consequently, a nonresonant diagonal scheme exists for all twist-$2$ operators in this theory, wherein the renormalized mixing matrices $Z(\lambda)$ from Eq. \eqref{ZZ} are one-loop exact with eigenvalues \cite{Bochicchio:2021geometry}
	\begin{equation}
		Z_{\mathcal{O}_i}(\lambda) = \Bigg(\frac{g(\mu)}{g(\frac{\mu}{\lambda})}\Bigg)^{\frac{\gamma_{0\mathcal{O}_i}}{\beta_0}} \,.
	\end{equation}
	It can be concluded that the UV asymptotics of the generating functional for correlators, as computed in \cite{BPS41,BPS42}, is applicable to all twist-$2$ operators within SUSY $\mathcal{N}=1$ SU($N$) Yang-Mills theory. 
	\section*{Acknowledgments}
	I wish to thank Marco Bochicchio for his review of the manuscript.

	\appendix
	
	\section{Bounds on sums}
	
	\subsection{Integral sandwich lemma}
	
	\begin{lemma}\label{lem:integral}
	
	Let $f:[a-1,b+1]\to\mathbb{R}$ be decreasing and $a,b\in\mathbb{Z}$ with $a\le b$. Then
	\begin{equation}\label{eq:sandwich}
	\int_{a}^{\,b+1} f(x)\,dx \le \sum_{k=a}^{b} f(k) \le \int_{a-1}^{\,b} f(x)\,dx.
	\end{equation}
	\end{lemma}
	
	\begin{proof}
	For each $k\in\{a,\dots,b\}$ and $x\in[k,k+1]$, decreasingness gives $f(x)\le f(k)$, from which we get 
	\begin{equation}
	\int_{k}^{k+1} f(x)\,dx \le f(k)\int_{k}^{k+1} \,dx\le f(k).
	\end{equation}
	Summing over $k=a,\dots,b$ yields the left inequality in \eqref{eq:sandwich}.
	Similarly, for $x\in[k-1,k]$ we have $f(x)\ge f(k)$, so
	\begin{equation}
	\int_{k-1}^{k} f(x)\,dx \ge f(k)\,,
	\end{equation}
	and summing over $k=a,\dots,b$ gives the right inequality.
	\end{proof}
	
	\begin{clr}\label{cor:harmonic-blocks}
	
	Taking $f(x)=1/x$ (decreasing on $(0,\infty)$) in \eqref{eq:sandwich} gives, for integers $1\le a\le b$,
	\begin{equation}\label{eq:harmonic-sandwich}
	\log\frac{b+1}{a}\le\sum_{k=a}^{b}\frac{1}{k}\le\log\frac{b}{a-1}.
	\end{equation}
	\end{clr}
	
	\begin{proof}
	Integrate $\frac{1}{x}$ to obtain $\int_{u}^{v} \frac{dx}{x}=\log v-\log u$.
	\end{proof}

	\subsection{Upper bound for $\Sigma_m(x)$ with $x<m$}
	
	\begin{prop}\label{prop:Sigma-main}
	For all $m\ge2$ and $x<m$
	\begin{equation}
			\Sigma_m(x) = \sum_{k = m+1}^{m+x}\frac{1}{k}< \log 2,
	\end{equation}
	\end{prop}
	
	\begin{proof}
	From Eq. \eqref{eq:harmonic-sandwich},
	\begin{equation}
	\Sigma_m(x) \le \log\left(1+\frac{x}{m}\right)
	\le \log 2 
	\end{equation}
	since $x<m$. 
	\end{proof}
	\subsubsection{An alternative route}
	By noticing that the anomalous dimensions are all monotonically increasing functions we immediately establish that for $x<m$
	\begin{equation}
\Sigma_m(x)<\Sigma_m(m-1)\,.
	\end{equation}
	Now, as it was noticed in \cite{S1}, we show that also $\Sigma_m(m-1)$ is montonic in $m$
	\begin{align}
\Sigma_{m+1}(m)-\Sigma_m(m-1) &=  \sum_{k = m+2}^{2m+1}\frac{1}{k}- \sum_{k = m+1}^{2m-1}\frac{1}{k}\nonumber\\
&=\frac{1}{2m+1}+\frac{1}{2m}-\frac{1}{m+1}\nonumber\\
&=\frac{3 m+1}{4 m^3+6 m^2+2 m}>0
	\end{align}
	Therefore we can use the bound \cite{S1}
	\begin{equation}
\Sigma_m(m-1)\le\lim_{m\to\infty}\Sigma_m(m-1)=\log 2
	\end{equation}
	\subsection{Upper bound for $K_m(x)$ with $x<m$}

\begin{prop}\label{prop:K}
For all $m\ge 2$
\begin{equation}
K_m(x)\le2\log 2.
\end{equation}
\end{prop}

\begin{proof}
We start from the definition of $K_m(x)$
\begin{equation}
K_m(x) =2\sum_{k = m-1}^{m-2+x}\frac{1}{k}+\frac{3}{m+x}-\frac{3}{m}
\end{equation}
then by applying  Eq. \eqref{eq:harmonic-sandwich} on the first term 
\begin{equation}
K_m(x)\le 2\log\left(1+\frac{x}{m-2}\right)+\frac{3}{m+x}-\frac{3}{m}\,,
\end{equation}
since $1\le x<m$ we can easily bound this sum as
\begin{equation}
\frac{3}{m+x}-\frac{3}{m}\le 0
\end{equation}
and so
\begin{equation}
K_m(x)\le 2\log 2
\end{equation}
thus completing the proof.
\end{proof}

\subsection{Upper bound for $J_m(x)$ for $x<m$}

\begin{prop}\label{prop:J-upper}
For all $m\ge2$ and $x<m$
\begin{equation}
J_m(x)<2\log 2.
\end{equation}
\end{prop}

\begin{proof}
From the definition of $J_m(x)$
\begin{equation}
	J_m(x) =2\sum_{k = m+1}^{m+x}\frac{1}{k}+\frac{2(-1)^{m+x}}{(m+x+2)(m+x+1)(m+x)}-\frac{2(-1)^m}{(m+2)(m+1)m}
\end{equation}
we use as above Eq. \eqref{eq:harmonic-sandwich} on the first term 
\begin{equation}
J_m(x) \le 2\log\left(1+\frac{x}{m}\right)+\frac{2(-1)^{m+x}}{(m+x+2)(m+x+1)(m+x)}-\frac{2(-1)^m}{(m+2)(m+1)m}
\end{equation}
we use the fact that for $x<m$
\begin{equation}
\log\Bigl(1+\frac{x}{m}\Bigr)\ \le\ \log\Bigl(1+\frac{m-1}{m}\Bigr)
=\log\Bigl(2-\frac{1}{m}\Bigr)
=\log 2+\log\!\Bigl(1-\frac{1}{2m}\Bigr).
\end{equation}
Using $\log(1-u)\le -u$ for $u\in(0,1)$,
\begin{equation}\label{eq:stima2}
2\log\!\Bigl(1-\frac{x}{2m}\Bigr)\ \le\ 2\log 2-\frac{1}{m}.
\end{equation}
For the rational terms we have that in the worst case scenario both $-(-1)^m$ and $(-1)^{m+x}$ yield a positive term and so we consider
\begin{align}
J_m(x) &\le 2\log\left(1-\frac{x}{m}\right)+\frac{2}{(m+x+2)(m+x+1)(m+x)}+\frac{2}{(m+2)(m+1)m}\nonumber\\
&\le 2\log 2-\frac{1}{m}+\frac{3}{m \left(m^2+5 m+6\right)}\nonumber\\
&= 2\log 2-\frac{m^2+5 m+3}{m(m^2+5 m+6)}\nonumber\\
&\le 2\log 2
\end{align}
where in the second line we have put $x=1$ as a majorant and we have also used Eq. \eqref{eq:stima2}.
\end{proof}

\subsection{Upper bound for $L_m(x)$ with $x<m$}

\begin{prop}\label{prop:L}
For all $m\ge 3$ and $x<m$
\begin{equation}\label{eq:Lm-upper-sharp}
L_m(x)\le2\log 2+\frac{3}{2}.
\end{equation}
\end{prop}

\begin{proof}
From the definition of $L_m(x)$
\begin{equation}
L_m(x)=2\sum_{k = m-1}^{m+x-2}\frac{1}{k}+\frac{2(-1)^{m+x-1}}{(m+x)(m+x-1)}-\frac{2(-1)^{m-1}}{m(m-1)}
\end{equation}
then by applying  Eq. \eqref{eq:harmonic-sandwich} on the first term 
\begin{equation}
L_m(x)\le 2\log\left(1+\frac{x}{m-2}\right)+\frac{2(-1)^{m+x-1}}{(m+x)(m+x-1)}-\frac{2(-1)^{m-1}}{m(m-1)}
\end{equation}
we use the fact that for $x<m$
\begin{equation}
\log\!\Bigl(1+\frac{x}{m-2}\Bigr)\ \le\ \log\!\Bigl(1+\frac{m-1}{m-2}\Bigr)
=\log\!\Bigl(2+\frac{1}{m-2}\Bigr)
=\log 2+\log\!\Bigl(1+\frac{1}{2(m-2)}\Bigr).
\end{equation}
Using $\log(1+u)\le u$ for $u\in(0,1)$,
\begin{equation}\label{eq:stima2}
2\log\left(1+\frac{x}{m-2}\right)\ \le\ 2\log 2+\frac{1}{m-2}.
\end{equation}
For the rational terms we have that in the worst case scenario both $-(-1)^{m-1}$ and $(-1)^{m+x-1}$ yield a positive term and so we consider
\begin{align}
L_m(x)&\le 2\log 2+\frac{1}{m-2}+\frac{2}{(m+x)(m+x-1)}+\frac{2}{m(m-1)}\nonumber\\
&\le 2\log 2+\frac{1}{m-2}+\frac{2}{(m+1)m}+\frac{2}{m(m-1)}\nonumber\\
&= 2\log 2+\frac{m (m+4)-9}{(m-2) \left(m^2-1\right)}\nonumber\\
&\le 2\log 2+\frac{3}{2}
\end{align}
\end{proof}

	\subsection{ Upper bound for $U_m(x)$ for $x<m$}

	\begin{prop}\label{prop:U}
	For all $m\ge2$ and $x<m$
	\begin{equation}
U_m(x)<2\log 2.
	\end{equation}
	\end{prop}
	\begin{proof}
	From the definition of $U_m(x)$
	\begin{equation}
	U_m(x) =2\sum_{k = m+2}^{m+x+1}\frac{1}{k}-\frac{2(-1)^{m+x}}{(m+x+2)(m+x+1)} + \frac{2(-1)^m}{(m+2)(m+1)}
	\end{equation}
	we use as above Eq. \eqref{eq:harmonic-sandwich} on the first term 
	\begin{equation}
	U_m(x) \le 2\log\left(1+\frac{x}{m+1}\right)-\frac{2(-1)^{m+x}}{(m+x+2)(m+x+1)} + \frac{2(-1)^m}{(m+2)(m+1)}
	\end{equation}
	we use the fact that for $x<m$
	\begin{equation}
	\log\!\Bigl(1+\frac{x}{m+1}\Bigr)\ \le\ \log\!\Bigl(1+\frac{m-1}{m+1}\Bigr) \le\ \log\!\Bigl(1+\frac{m-1}{m}\Bigr)
	=\log\!\Bigl(2-\frac{1}{m}\Bigr)
	=\log 2+\log\!\Bigl(1-\frac{1}{2m}\Bigr).
	\end{equation}
	Using $\log(1-u)\le -u$ for $u\in(0,1)$,
	\begin{equation}\label{eq:stima2}
	2\log\!\Bigl(1-\frac{x}{m+1}\Bigr)\ \le\ 2\log 2-\frac{1}{m}.
	\end{equation}
	For the rational terms we have that in the worst case scenario both $(-1)^m$ and $(-1)^{m+x}$ yield a positive term and so we consider
	\begin{align}
	U_m(x) &\le 2\log\left(1+\frac{x}{m+1}\right)+\frac{2}{(m+x+2)(m+x+1)} + \frac{2}{(m+2)(m+1)}\nonumber\\
	&\le 2\log 2-\frac{1}{m}+\frac{2}{m^2+4 m+3}\nonumber\\
	&= 2\log 2-\frac{m^2+2 m+3}{m^3+4 m^2+3 m}\nonumber\\
	&\le 2\log 2
	\end{align}
	where in the second line we have put $x=1$ as a majorant and we have also used Eq. \eqref{eq:stima2}.
	\end{proof}

	\bibliographystyle{JHEP}
	\bibliography{mybib} 
	
\end{document}